\documentclass[runningheads]{llncs}
\usepackage[utf8]{inputenc}
\usepackage[T1]{fontenc}
\usepackage{graphicx}
\usepackage{hyperref}
\usepackage{color}

\usepackage{algorithm}
\usepackage[noend]{algpseudocode}
\algrenewcommand\algorithmicfunction{}
\algrenewcommand\algorithmicprocedure{}

\usepackage{amsmath}
\usepackage{array}
\usepackage{verbatim}

\usepackage{todonotes}
\usepackage{amsmath}

\renewcommand{\epsilon}{\varepsilon}
\newcommand{\Bal}{\mathit{Sturm}}
\newcommand{\Pal}{\mathit{Pal}}
\newcommand{\Lyn}{\mathit{Lyn}}

\begin{document}
\title{
Dorst–Smeulders Coding for\\ Arbitrary Binary Words
\thanks{Gabriele Fici is partly supported by MUR project PRIN 2022 APML – 20229BCXNW, funded by the European Union – Mission 4 ``Education and Research'' C2 - Investment 1.1.}}
%
%
\author{
Alessandro De Luca\inst{2}\orcidID{0000-0003-1704-773X}\\ \and
Gabriele Fici\inst{3}\orcidID{0000-0002-3536-327X} }
\authorrunning{A. De Luca and G. Fici}
%
\institute{
DIETI, Università di Napoli Federico II, Italy\\
\email{alessandro.deluca@unina.it}\\ \and
Dipartimento di Matematica e Informatica, Università di Palermo, Italy\\
\email{gabriele.fici@unipa.it}}
\maketitle              
\begin{abstract}
 A binary word is Sturmian if the occurrences of each letter are balanced, in the sense that in any two factors of the same length, the difference between the number of occurrences of the same letter is at most $1$. In digital geometry, Sturmian words correspond to discrete approximations of straight line segments in the Euclidean plane. The Dorst–Smeulders coding, introduced in 1984, is a 4-tuple of integers that uniquely represents a Sturmian word $w$, enabling its reconstruction using $|w|$ modular operations, making it highly efficient in practice. In this paper, we present a linear-time algorithm that, given a binary input word $w$, computes the Dorst–Smeulders coding of its longest Sturmian prefix. This forms the basis for computing the Dorst–Smeulders coding of an arbitrary binary word
$w$, which is a minimal decomposition (in terms of the number of factors) of $w$ into Sturmian words, each represented by its Dorst–Smeulders coding. This coding could be leveraged in compression schemes where the input is transformed into a binary word composed of long Sturmian segments. Although the algorithm is conceptually simple and can be implemented in just a few lines of code, it is grounded in a deep analysis of the structural properties of Sturmian words.
\end{abstract}

\section{Introduction}

Sturmian (finite) words are binary balanced words. They can be used to describe straight line segments in the discrete plane. In 1984, Dorst and Smeulders introduced a coding for uniquely representing a Sturmian word~\cite{DBLP:journals/pami/DorstS84}. Although other codings were proposed in the literature (see, e.g.,~\cite{DBLP:journals/pami/LindenbaumK91}) the coding of Dorst and Smeulders has a clear interpretation in terms of combinatorics on words. In fact, Sturmian words are factors of (lower primitive) Christoffel words, and Christoffel words have lots of interesting combinatorial properties. For every length $\ell>0$ and every height (number of $1$'s) $h$ coprime with $\ell$, there is exactly one Christoffel word $u_{h,\ell}$. Moreover, Christoffel words are precisely the Lyndon Sturmian words. Let $w$ be a Sturmian word of minimum period $p$. The Lyndon conjugate of the prefix of length $p$ of $w$ is therefore a Christoffel word $u_{h,p}$, and $w_1\cdots w_p=\sigma^s(u_{h,p})$ for some $s\geq 0$, where $\sigma$ is the usual (right) shift operator. The word $w$ is then uniquely determined by its length $n$, its period $p$, the height $h$, and the shift $s$. The $4$-tuple $(n,p,h,s)$ is the \emph{Dorst–Smeulders coding} of $w$. 
In this paper, we extend the Dorst–Smeulders coding to an arbitrary binary word $w$ by considering a particular factorization of $w$ in Sturmian words and encoding each term with its Dorst–Smeulders coding.

Let $\Sigma$ be an alphabet, and let $L\subseteq \Sigma^*$ be a language such that $\Sigma\subseteq L$. Then every word $w$ in $\Sigma^*$ can be factored in words of $L$, i.e., there exist $x_1,\ldots,x_k\in L$ such that $w=x_1\cdots x_k$. We call $k$ the \emph{length} of the factorization. We call a factorization  $w=x_1\cdots x_k$ \emph{minimal} if any other factorization in words of $L$ has length at least $k$.
For example, if $L=\Pal$ is the language of palindromes, the length of a minimal factorization of $w$ is also called the \emph{palindromic length} of $w$~\cite{DBLP:journals/aam/FridPZ13}. 

For some languages, a minimal factorization can be obtained by the \emph{greedy algorithm}, which consists in taking the longest prefix of $w$ that belongs to $L$ and recursing on what remains after removing this prefix from $w$. 
This is the case, for example, for the language $\Lyn$ of Lyndon words, where the minimal factorization is also called the \emph{Chen--Fox--Lyndon factorization}~\cite{chen1958free}, but it is not the case for the language of palindromes. For example, the greedy factorization of $w=0010$ is $w=00\cdot 1\cdot 0$ and it is not minimal, since $w=0\cdot 010$.

In particular, if the language $L$ is \emph{factorial} (i.e., it is closed under taking factors), then the greedy algorithm always produces a minimal factorization---this can be easily proved by contradiction. Actually, it is sufficient that the language is closed under taking suffixes~\cite{DBLP:conf/dcc/CohnK96}.
By symmetry, for languages closed under taking prefixes, the right-to-left greedy factorization produces a minimal factorization.

The complexity of the problem of determining a minimal factorization depends on the language $L$. For example, it is known that this complexity is linear in the length of the input word for $\Pal$~\cite{kosolobov2017palindromic} and $\Lyn$~\cite{DBLP:journals/jal/Duval83}. 

For those languages such that the greedy algorithm produces a minimal factorization, this complexity reduces to the complexity of computing the longest prefix of $w$ that belongs to $L$. 

In this paper, we are particularly interested in the language $\Bal$ of Sturmian words. A binary word is Sturmian if the occurrences of each letter are balanced, in the sense that in any two factors of the same length, the difference between the number of occurrences of the same letter is at most $1$. The language $\Bal$ is a factorial language; hence, a minimal Sturmian factorization can be obtained by the greedy algorithm. 
We then call \emph{Dorst–Smeulders coding} of an arbitrary binary word $w$ the list of  Dorst–Smeulders codings of the terms in the minimal Sturmian factorization of $w$ obtained from the greedy algorithm.  
We present a linear-time algorithm to compute the Dorst–Smeulders coding of an arbitrary binary word. The core of the algorithm is a procedure that, given an arbitrary binary input word $w$, computes the longest prefix of $w$ that is Sturmian (a different linear time algorithm for this, based on geometric considerations, is described in~\cite{DBLP:journals/dm/BerstelP96}). Our procedure scans $w$ from left to right, letter by letter, checking for an arithmetic property that ensures balance, and updates the $4$ parameters of the Dorst–Smeulders coding, thus yielding as output the Dorst–Smeulders coding of the longest Sturmian prefix.
Our algorithm is straightforward to implement since it consists of only elementary arithmetic operations. 

\section{Preliminaries}\label{sec:prelim}

Let $\Sigma$ be a finite alphabet. A \emph{language} $L$ over $\Sigma$ is a set of words over $\Sigma$, i.e., a subset of $\Sigma^*$, the free monoid generated by $\Sigma$. A language $L$ is \emph{factorial} (or factor-closed) if it contains all the factors of its words.

Let $w=w_1w_2\cdots w_n$, $w_i\in \Sigma$, be a word of length $n=|w|$. An integer $p>0$ is a \emph{period} of $w$ if $w_i=w_j$ whenever $i=j \mod p$. The prefix $\rho(w)$ of $w$ whose length is the minimum period of $w$ is called the \emph{fractional root} of $w$.  A \emph{border} of $w$ is a factor that occurs as a prefix and as a suffix in $w$. The word $w$ has a border of length $b$ if and only if $|w|-b$ is a period of $w$. An \emph{unbordered word} is a word that coincides with its fractional root, i.e., such that $|w|$ is its minimum period.

The \emph{reversal} $\widetilde{w}$ of $w$ is the word $\widetilde{w}=w_nw_{n-1}\cdots w_1$. A word is a \emph{palindrome} if it coincides with its reversal.

Let  $w=w_1\cdots w_{n-1}w_{n}$ be a word of length $n>0$. The \textit{shift} of $w$ is the word $\sigma(w)=w_{n}w_1\cdots w_{n-1}$. Two words $w$ and $w'$ are \emph{conjugates} if $w=uv$ and $w'=vu$ for some words $u$ and $v$. The conjugacy class of a word $w$ can be obtained by repeatedly applying the shift operator, and contains $|w|$ distinct elements if and only if $w$ is \emph{primitive}, i.e., $w\neq v^k$ for any nonempty word $v$ and $k>1$.

A nonempty word is \emph{Lyndon} if it is lexicographically smaller than all its nonempty suffixes. Lyndon words are unbordered. Every primitive word $w$ has a (unique) Lyndon conjugate, i.e., $w=\sigma^s(w')$ for a Lyndon word $w'$ and an integer $s\geq 0$.

From now on, we suppose $\Sigma=\{0,1\}$, and we use the lexicographic order on $\Sigma^*$ induced by $0<1$. Occasionally, we will also consider the letters of $\Sigma$ as numbers, i.e., make use of their arithmetic value. 

A  word $w\in \Sigma^*$ is Sturmian (or $1$-balanced) if for every two factors $u$ and $v$ of $w$ of the same length, one has $||u|_0-|v|_0|\leq 1$ (or, equivalently, $||u|_1-|v|_1|\leq 1$), where $|w|_x$ denotes the number of occurrences of the letter $x$ in the word $w$. For example, $01001$, $010101$ and $110101$ are Sturmian words, whereas $0011$ is not. Famous examples of Sturmian words are the \emph{Fibonacci words}, defined recursively by $f_1=1$, $f_2=0$ and $f_n=f_{n-1}f_{n-2}$ for each $n>2$.

For any Sturmian word $w$, at least one between $0w$ and $1w$ (resp.~between $w0$ and $w1$) is Sturmian.
A Sturmian word $w$ is \emph{left (resp.~right) special} if $0w$ and $1w$ (resp.~$w0$ and $w1$) are both Sturmian, and it is \emph{bispecial} if it is left and right special.
Of course, a word is left (resp.~right) special if and only if it is a prefix (resp.~a suffix) of a bispecial word.
A bispecial word is \emph{strictly bispecial} if $0w0$, $1w0$, $0w1$ and $1w1$ are in $\Bal$. For example, $00$ is strictly bispecial; $10$, instead, is bispecial but not strictly bispecial, since $1\cdot 10 \cdot 0$ is not Sturmian.

A \emph{central word} is a word having two coprime periods, $p$ and $q$, and length $p+q-2$. It is well-known that central words are Sturmian and palindromes. 

\begin{proposition}[\cite{de1994some}]\label{central}
Let $w$ be a binary word. The following are equivalent:
\begin{enumerate}
    \item $w$ is a central word;
     \item $w$ is a strictly bispecial Sturmian word;
    \item $w$ is a power of a single letter or there exist palindromes (actually, central words) $P,Q$ such that $w=P01Q=Q10P$.
\end{enumerate}
\end{proposition}

A (lower) \emph{Christoffel word} is either a single letter or a word of the form $0c1$, where $c$ is a central word. Christoffel words are precisely the Lyndon Sturmian words.

Every Christoffel word $u=0c1$ has a unique palindromic factorization $u=\alpha\beta$. Moreover, the lengths of $\alpha$ and $\beta$ are the two coprime periods of $c$, as well as the multiplicative inverses of $|u|_0$ and $|u|_1$ modulo $|u|$, respectively. This factorization is straightforward if $c$ is a power of a single letter; otherwise, by Proposition~\ref{central}, one has $\alpha=0P0$ and $\beta=1Q1$.

For example, $010010010$ is a central word, with coprime periods $p=3$ and $q=8$, and $00100100\cdot 101$ is the palindromic factorization of the corresponding Christoffel word.

For more details on Christoffel and Sturmian words, the reader is pointed to \cite[Chap.~2]{LothaireAlg} and \cite{Book08}.

\section{The coding of Dorst and Smeulders}

We now describe how to code any Sturmian word using $4$ integers of size bounded by the length of the word. This coding is due to Dorst and Smeulders~\cite{DBLP:journals/pami/DorstS84}.

Every Sturmian word $w$  occurs as a factor of length $n=|w|$ starting at some position $s'$ in the infinite periodic Christoffel word $u_{h,p}^\omega=u_{h,p}u_{h,p}\cdots$ of slope $h/p$, where $p=|\rho(w)|$ is the minimum period of $w$ and $h=|\rho(w)|_1$ is the height of the fractional root of $w$. The number $s=(1-s')\bmod p$, called shift, is in fact the distance of the root of $w$ from its  Christoffel conjugate. In other words, the root $w_1w_2\cdots w_p$ of $w$ has  Christoffel conjugate $w_{s+1}w_{s+2}\cdots w_{s}$. In particular, $s=0$ if and only if the root of $w$ is a  Christoffel word.

For example, $w=101001$ occurs in the infinite periodic Christoffel word  $(00101)^\omega$ of slope $2/5$, starting at position $3$. Notice that $w$ has the same minimum period $p=5$ of $00101$, since $00101$ is a conjugate of the fractional root $10100$ of $w$, and the height $h=2$ of $00101$ is the height of the fractional root of $w$.

Therefore, every Sturmian word $w$ is completely determined by its length $n$, its period $p$, the height $h$ of its fractional root, and the shift $s$. The $4$-tuple $(n,p,h,s)$ is called the \textit{Dorst–Smeulders coding} of $w$.

Since for every $i\geq 1$, the $i$th letter of the infinite periodic Christoffel word $u_{h,p}^\omega$  is 
\[\left\lfloor i\frac{h}{p} \right\rfloor - \left\lfloor (i-1) \frac{h}{p}\right\rfloor\]
(or equivalently, the prefix of $u_{h,p}^\omega$ of length $i$ has height $\lfloor ih/p\rfloor$),
we have that the $i$th letter of $w$, for every $1\leq i\leq n$, is
\[w_i=\left\lfloor (i-s)\frac{h}{p} \right\rfloor - \left\lfloor (i-s-1) \frac{h}{p}\right\rfloor .\]

This leads to a reconstruction of $w$ from its Dorst–Smeulders coding that is very fast in practice, since it performs only arithmetic operations.

\section{An optimal online algorithm}

We now give a linear-time algorithm that computes the longest Sturmian prefix $w$ of an input word and returns the Dorst–Smeulders coding of $w$.

 It scans the input word from left to right, character by character, and maintains four integer variables, which represent the Dorst–Smeulders coding of the current balanced prefix. In order to check whether the next prefix is balanced, we make use of the following result. 

\begin{lemma}
\label{lem:rsp}
Let $w\neq\varepsilon$ be a Sturmian word with
Dorst–Smeulders coding $(n,p,h,s)$, and $a\in\{0,1\}$.
If $wa$ has period $p$, then it is Sturmian.
Otherwise, $wa$ is Sturmian (and $w$ is right special) if and only if 
\begin{equation}\label{eq:rsp}
    (-h(n+1-s))\bmod p\in\{0,1\}.
\end{equation}
\end{lemma}

\begin{proof}
Since $w$ is Sturmian if and only if its fractional root $\rho(w)$ is a conjugate of a Christoffel word (cf.~\cite{DBLP:journals/tcs/LucaL06a}), clearly $wa$ is balanced if $|\rho(wa)|=p$, that is, if $\rho(wa)=\rho(w)$.

On the other hand, if $wa$ does not have period $p$, then necessarily $wb$ does, where $\{a,b\}=\{0,1\}$. Thus, in this case $wa$ is balanced if and only if $w$ is right special, which in turn is equivalent to $\widetilde w$ being \emph{left} special. Now, a Sturmian word is left special if and only if its fractional root is either a single letter or $cxy$ for a central word $c$ and $\{x,y\}=\{0,1\}$~\cite{de2006pseudopalindrome}.
The single letter case means $p=1$, trivially satisfying~\eqref{eq:rsp}. Otherwise, either $10c$ or $01c$ is a suffix of $w$.

The first option holds if and only if $w=\lambda u^k0c$, where $k\geq 0$ and $\lambda$ is the suffix of $u:=u_{h,p}=0c1$ of length $s$ (recall that the first occurrence of the lower Christoffel word $u$ in $\rho(w)^\omega$ begins at position $s+1$). Equivalently, considering the lengths of the words involved, we have $n+1-s\equiv 0\pmod p$; as $p>1$ and $\gcd(h,p)=1$, this is equivalent to $-h(n+1-s)\equiv 0\pmod p$.

Finally, let us examine the second possibility, i.e., $w$ ending in $01c$. Recall that $u=0c1$, being a primitive Christoffel word of length $p>1$, has a unique factorization $u=\alpha\beta$ in two palindromes $\alpha,\beta$, whose lengths verify~(see~\cite{berthe2008involution})
\begin{equation}
\label{eq:dual}
|\alpha|\cdot|u|_0\equiv |\beta|\cdot|u|_1\equiv 1\pmod{p}.
\end{equation}
Now, $w$ ends in $01c$ if and only if $w0=\lambda u^m\alpha$ for some $m\geq 0$. Since we have $|u|_1=|\rho(w)|_1=h$ and $|\beta|\equiv -|\alpha|\pmod p$, the above equation is equivalent to $-h|\alpha|\equiv 1\pmod p$, and hence to $-h(n+1-s)\equiv 1\pmod p$, as $n+1-s=|u^m\alpha|=mp+|\alpha|\equiv|\alpha|\pmod p$.
\qed
\end{proof}

\begin{remark}
    The same ideas from the previous lemma lead to a characterization of Sturmian words in terms of their \emph{period array} (i.e., the array whose $i$-th entry is the minimum period of the prefix of length $i$, for $1\leq i\leq n$; note the connection with the \emph{border} array). Namely, a binary word $w$ with period array \[(\underbrace{p_1,\ldots,p_1}_{k_1},\underbrace{p_2,\ldots,p_2}_{k_2},\ldots,\underbrace{p_m,\ldots,p_m}_{k_m})\]
    (with $1=p_1<p_2<\cdots<p_m=|\rho(w)|$) is balanced if and only if $p_j$ divides either $k_j$ or $k_j+p_{j-1}$, whenever $1<j<m$.
\end{remark}

\begin{lemma}
\label{lem:update}
    Let $w\in\{0,1\}^*$ be a right special Sturmian word and  $(n,p,h,s)$ be its Dorst–Smeulders coding, with $p>1$. If $a\in\{0,1\}$ is such that $|\rho(wa)|>p$, then the coding $(n+1,p',h',s')$ of $wa$ satisfies the following:
    \begin{itemize}
        \item $p'=n+1-((n+1-(-1)^a\bar h)\bmod p)$,
        \item $h'=\left\lfloor\frac{n+1-(-1)^a\bar h}{p}\right\rfloor h+(-1)^a\left\lfloor\frac {h\bar h}{p}\right\rfloor$,
        \item $s'=as + (1-a)(n+1-p)$,
    \end{itemize}
    where $\bar h$ is the multiplicative inverse of $h$ modulo $p$.
\end{lemma}

\begin{proof}
    By Lemma~\ref{lem:rsp}, the number $-h(n+1-s)$ is, modulo $p$,  either $0$ or $1$.
    \begin{enumerate}
    \item As seen in the proof of Lemma~\ref{lem:rsp},
    the first case occurs when $w=\lambda(0c1)^k0c$ for some lower Christoffel word $u=0c1$ of length $p$, and $\lambda$ its proper suffix of length $s$. As $|\rho(wa)|\neq p$, we have $a=0$ in this case.
    
    Now, $0c0$ occurs only as a suffix in $wa=w0$, and all other factors of length $p$ are conjugates of the Lyndon word $0c1$, thus lexicographically greater. Therefore,
    $0c0$ is the lexicographically least factor of length $p$, so that the Lyndon conjugate of the root $\rho(wa)$ (whose length is larger than $p$, by hypothesis) must start with $0c0$. In other words, we have \[s'=|\lambda(0c1)^k|=|w0|-|0c0|=n+1-p,\] the desired value for $a=0$. Now let $0c1=\alpha\beta$ be the palindromic factorization. The unique occurrence of $0c0$ also implies that the longest border of $wa$ is the longest proper suffix of $0c0$ (or equivalently, of $1c0=\beta\alpha$) that is also a prefix of $\lambda 0c$.

    If $s<|\beta|$, then the word $\lambda$, being a suffix of $0c1=\alpha\beta$ in general, must be a suffix of $\beta$. Hence, by the above observation, the longest border of $wa=\lambda (\alpha\beta)^k 0c0$ is also the longest border of $\lambda\alpha\beta\alpha$, which is $\lambda\alpha$ (as $\beta\alpha$ is unbordered). Thus, $p'=|wa|-|\lambda\alpha|=n+1-(s+|\alpha|)$; this is the desired value, since in this case we have $a=0$, $n+1\equiv s\pmod p$, and $\bar h=|\beta|\equiv -|\alpha|\pmod p$ by equation~\eqref{eq:dual}. We also have
    \[\begin{split}h' &=|\beta(\alpha\beta)^{k-1}0c0|_1=|\beta (0c1)^k|_1-1=|\beta|_1-1+kh\\ &=\left\lfloor\frac{h\bar h}{p}\right\rfloor+\left(\frac{n+1-s}{p}-1\right)h,\end{split}\]
    where for the last equality we used the fact that $\beta$ is the suffix of length $\bar h$ of the Christoffel word $0c1$, and so its height minus 1 equals the height of the \emph{prefix} of $0c1$ of the same length (recall that $c$ is a palindrome). The value of $h'$ again satisfies our thesis since $s<\bar h<p$.

    If $s\geq|\beta|=\bar h$, instead, let $\lambda=\gamma\beta$, so that $\gamma$ is a proper suffix of $\alpha$. The longest border of $wa$ then coincides with the longest border of $\lambda\alpha=\gamma\beta\alpha$, which is $\gamma$. Thus, we again obtain $p'=n+1-((n+1-\bar h) \bmod p)$. Moreover, we have
    \[h'=|\beta(0c1)^{k+1}|_1-1=|\beta|_1-1+(k+1)h=\left\lfloor\frac{h\bar h}{p}\right\rfloor+\frac{n+1-s}{p}h,\] again as desired, since $0<\bar h<s<p$ here.
    
    \item By Lemma~\ref{lem:rsp}, the second case $-h(n+1-s)\equiv 1\pmod p$ occurs when $w0=\lambda(\alpha\beta)^k\alpha$ but $a=1$, so that $1c1$ occurs only as a suffix in $wa=w1$. Therefore, the longest border of $wa$ is the longest proper suffix of $1c1$ (and hence also of $0c1=\alpha\beta$) that is also a prefix of $\lambda\alpha\beta$; in other words, it coincides with the longest border of $\lambda\alpha\beta$, which is $\lambda$. Thus, $p'=n+1-s$ in this case; as $n+1-s+\bar h=|(\alpha\beta)^{k+1}|\equiv 0\pmod p$, the formula in the statement is again verified. Now, the lower Christoffel conjugate of $\rho(wa)$ must end with the factor of length $p$ that is lexicographically greatest, that is, $1c1$. This implies $s'=s$.

    Finally, we have \[h'=|(\alpha\beta)^k\alpha|_1+1=|(\alpha\beta)^{k+1}|_1-(|\beta|_1-1)=\left\lfloor\frac{n+1+\bar h}{p}\right\rfloor h-\left\lfloor\frac{h\bar h}{p}\right\rfloor.\quad\qed\]
    \end{enumerate}
\end{proof}

\begin{algorithm}
\caption{Dorst--Smeulders coding of the longest Sturmian prefix\label{DS}}
\begin{algorithmic}[1]
\Require Binary word $w = w_0 w_1 \dots w_{N-1}$, where $w_i \in \{0,1\}$
\Ensure Tuple $(n, p, h, s)$ that is the Dorst--Smeulders coding of the longest Sturmian prefix of $w$
\If{all letters in $w$ are equal}
    \State \Return $(N, 1, w_0, 0)$
\EndIf
\State $p \gets$ index of first occurrence of $1 - w_0$, plus $1$
\State $h \gets (-1)^{w_0} \bmod p$
\State $s \gets (-w_0) \bmod p$
\For{$n = p$ \textbf{to} $N - 1$}
    \If{$w_n \ne w_{n - p}$}
        \State $h^{-1} \gets$ modular inverse of $h$ modulo $p$
        \If{$(n - s + 1) \bmod p = 0$}
            \State $h \gets \left\lfloor \frac{n+1 - h^{-1}}{p} \right\rfloor \cdot h + \left\lfloor \frac{h \cdot h^{-1}}{p} \right\rfloor$
            \State $s \gets n + 1 - p$
            \State $p \gets n + 1 - ((n + 1 - h^{-1}) \bmod p)$
        \ElsIf{$(n - s + 1 + h^{-1}) \bmod p = 0$}
            \State $h \gets \left\lfloor \frac{n+1 + h^{-1}}{p} \right\rfloor \cdot h - \left\lfloor \frac{h \cdot h^{-1}}{p} \right\rfloor$
            \State $p \gets n + 1 - ((n + 1 + h^{-1}) \bmod p)$
        \Else
            \State \Return $(n, p, h, s)$
        \EndIf
    \EndIf
\EndFor
\State \Return $(N, p, h, s)$
\end{algorithmic}
\end{algorithm}

The total running time is clearly linear, and the working space is constant, assuming the Word-RAM model. The pseudocode is shown in Algorithm~\ref{DS}.


\begin{example}
Let 
$w=0101001101010000010010010101001001000101$, of length $40$. The Dorst–Smeulders coding of $w$ is
\[(7, 5, 2, 4), (7, 7, 3, 5), (11, 10, 3, 0), (11, 11, 4, 3), (4, 2, 1, 0)\]
corresponding to the factorization
\[0101001\cdot 1010100\cdot 00010010010\cdot 10100100100\cdot 0101.\]
\end{example}

\section{Conclusions and future work}

We described a simple algorithm that, on a given binary input string, returns the Dorst–Smeulders coding of the Sturmian factors in the greedy decomposition. Our algorithm is conceptually simple and needs only constant working space in the Word-RAM model, since it only needs to maintain a constant number of integers whose size is bounded by $n$.
Our algorithm could be used as part of a compression scheme after applying some preprocessing that maps repetitive strings to strings with long Sturmian factors.


\end{document}